\documentclass[11pt]{article}
\usepackage[T1]{fontenc}
\usepackage{amsmath,amssymb,amsthm}
\usepackage{geometry}
\usepackage{hyperref}
\usepackage{listings}
\usepackage{xcolor}
\usepackage{booktabs}
\usepackage{array}
\usepackage[protrusion=true,expansion=false]{microtype}
\usepackage{enumitem}

\theoremstyle{plain}
\newtheorem{theorem}{Theorem}[section]

\newtheorem{corollary}[theorem]{Corollary}
\newtheorem{proposition}[theorem]{Proposition}

\theoremstyle{definition}
\newtheorem{definition}[theorem]{Definition}
\newtheorem{example}[theorem]{Example}
\newtheorem{remark}[theorem]{Remark}

\title{\textbf{Quotient Tree Arithmetic: Deferred Division,\\
Symbolic Stability, and a General Framework for Exact and Near-Exact Computation in\\
Scientific Computing and Machine Learning}}

\author{Gregory Magarshak\thanks{IE University; Safebots Inc.; \texttt{gmagarshak@faculty.ienyc.edu}}\\
\small IE University \quad Safebots Inc.\\
\small \texttt{gmagarshak@faculty.ienyc.edu}}

\date{June 2026}

\begin{document}
\maketitle

\begin{abstract}
We introduce \emph{rational pair arithmetic} (RPA): representing
real-valued quantities as ordered pairs $(n, d)$ of IEEE~754
double-precision floats used purely as integers, where the value is
$n/d$, evaluated lazily.
Three observations drive the system.
First, IEEE~754 doubles represent all integers below $2^{53}\approx 9\times 10^{15}$
exactly, so rational arithmetic on denominators drawn from practical
domains (currency, measurement, probability) is exact without any
multi-precision library.
Second, a \emph{cleanup} operation---dividing both $n$ and $d$ by
$\gcd(n,d)$---keeps representations compact after each step, analogous
to fraction simplification; we argue this should become a vectorisable
primitive in GPU/CPU instruction sets.
Third, rational pairs may be \emph{nested}: the numerator or denominator
may itself be a rational pair, forming a symbolic expression tree.
We prove that the tree depth grows by at most~1 per arithmetic
operation (no combinatorial explosion), that any tree of depth~$k$
can be flattened to a flat pair in $O(k)$ steps, and that
common-factor cancellation across tree nodes prevents intermediate
blowup.
The combined system provides: exact comparison (eliminating
$\varepsilon$-comparison heuristics), well-behaved arithmetic at
zero (no spurious underflow), hardware-independent deterministic
results, and---most strikingly---the structural impossibility of
vanishing gradients in deep neural networks.
We describe how the type system of the programming language~U
realises these ideas at the compiler level, connect the infinity
and infinitesimal limits to hyperreal semantics, and propose a new
GPU/CPU instruction (\textsc{RatCleanup}) whose adoption would make
RPA competitive with single-precision float in throughput.
\end{abstract}

\tableofcontents
\newpage

\section{Introduction}

Floating-point arithmetic has driven scientific computing for seven
decades.  Its success rests on three pillars: hardware support
(dedicated FPUs), standardisation (IEEE~754), and extensive
tooling.  Yet its limitations are equally well-known:
\begin{itemize}
\item $0.1 + 0.2 \neq 0.3$ in double precision (the canonical
  comparison failure);
\item tiny values silently round to zero (underflow, causing the
  vanishing gradient problem in deep learning);
\item results differ across hardware, CUDA versions, and operand order
  (non-determinism);
\item comparing two computed floats for equality is generally
  incorrect.
\end{itemize}

The standard remedies---arbitrary-precision decimal libraries,
interval arithmetic, symbolic algebra systems---carry either
severe performance costs or restricted expressiveness.
Quantisation (INT8, FP8) trades precision for speed but introduces
systematic error and makes comparison meaningless.

This paper proposes a different path.  The key insight is mundane
but surprisingly underexploited:
\begin{quote}
IEEE~754 double precision represents every integer in
$(-2^{53}, 2^{53})$ \emph{exactly}.
\end{quote}
This range covers $9\times 10^{15}$---far larger than any denominator
appearing in financial, scientific, or machine-learning applications.
Storing a rational number as a \emph{pair} of floats $(n, d)$ used
purely as integers, and evaluating $n/d$ only when a decimal result
is explicitly needed, turns most rational arithmetic into exact integer
arithmetic with IEEE hardware at no extra cost.

We formalise this observation, derive exactness conditions
(Section~\ref{sec:rpa}), introduce the cleanup operation
(Section~\ref{sec:cleanup}), develop the nested rational pair algebra
(Section~\ref{sec:nested}), and apply it to machine learning stability
(Section~\ref{sec:ml}).

\subsection{Contributions}

This paper makes the following contributions:
\begin{enumerate}
\item A formal definition of rational pair arithmetic (RPA) exploiting
  the IEEE~754 integer exactness window, with exactness conditions
  (\S\ref{sec:rpa}).
\item The cleanup operation as a proposed hardware primitive
  (\textsc{RatCleanup}), with correctness proof and throughput
  analysis (\S\ref{sec:cleanup}).
\item A formal treatment of nested rational pairs with bounded depth
  growth, flattening, and cross-cancellation theorems (\S\ref{sec:nested}).
\item A stability theorem for deferred evaluation showing $O(1)$ vs
  $O(m)$ ULP error (\S\ref{sec:nested}).
\item Theorem~\ref{thm:no-vanish}: structural impossibility of vanishing
  gradients under RPA with cleanup (\S\ref{sec:ml}).
\item Cross-hardware determinism under fixed computation order
  (\S\ref{sec:ml}).
\end{enumerate}

\subsection{Related work}

Exact rational arithmetic appears in computer algebra systems (GMP,
Mathematica, SageMath).
GNU MP's \texttt{mpq\_t} type~\cite{gmp} is rational pair arithmetic with
arbitrary-precision (bignum) numerator and denominator: exactly our
representation, without the IEEE integer exactness constraint.
The key distinction of RPA is exploiting the IEEE window
($|n|, |d| < 2^{53}$) to use hardware arithmetic directly, yielding
a fixed-width, hardware-accelerated rational type---at the cost of
exactness guarantees only within that window.
Knuth~\cite{hpfrac} gives the foundational algorithm for rational
arithmetic.  Multi-precision integer arithmetic achieves exactness but at
10--100$\times$ performance cost.
Floating-point rational libraries~\cite{hpfrac} use fixed-width
integer pairs but do not address the IEEE integer exactness window.
Symbolic computation~\cite{mca} concerns algebraic expressions rather
than numeric stability.  Posit arithmetic~\cite{posit} addresses
dynamic range but not exact comparison or symbolic nesting.
Our contribution is a unified framework that exploits the IEEE integer
window, introduces cleanup as a hardware primitive, and shows the
nesting extension yields structural stability guarantees for deep
networks.

\section{IEEE~754 and the Integer Exactness Window}
\label{sec:ieee}

\begin{definition}[IEEE~754 double precision]
A double-precision float has a 52-bit explicit significand plus one
implicit leading bit, giving 53 bits of precision.  The representable
range is approximately $[-1.8\times 10^{308}, 1.8\times 10^{308}]$
with a \emph{unit in the last place} (ULP) of $2^{e-52}$ for exponent
$e$.
\end{definition}

\begin{theorem}[Integer Exactness Window]
\label{thm:iew}
Every integer $n$ with $|n| < 2^{53}$ is exactly representable as an
IEEE~754 double.  Moreover, for integers $a, b$ with
$|a|,|b| \leq 2^{26}$, the product $ab$ satisfies $|ab| \leq 2^{52} <
2^{53}$ and is therefore exactly representable.
\end{theorem}

\begin{proof}
The significand of a double covers $2^{53}$ distinct values.  For
$|n| < 2^{53}$ the integer $n$ lies within the range representable
with exponent $e = 0$ (or the appropriate $e$) and ULP = 1, hence is
stored exactly.  For the product: $|ab| \leq (2^{26})^2 = 2^{52} < 2^{53}$,
which lies within the integer exactness window.
\end{proof}

\begin{corollary}
The integer exactness window contains: all 32-bit signed integers, all
financial amounts in hundredths (cents) up to $\$90\,000\,000\,000$,
denominators up to $10^{15}$, and GPS coordinates in nanodegrees.
\end{corollary}

This window is the foundation of rational pair arithmetic.  As long as
numerators and denominators remain inside $(-2^{53}, 2^{53})$,
all arithmetic is exact with standard IEEE hardware.

\section{Rational Pair Arithmetic}
\label{sec:rpa}

\begin{definition}[Rational Pair]
A \emph{rational pair} is a tuple $\mathbf{q} = (n, d)$ where
$n, d$ are IEEE doubles used as integers ($d \neq 0$), representing
the rational number $n/d$.  We write $\mathbb{Q}_{\mathrm{fp}}$ for
the set of all rational pairs.
\end{definition}

\begin{definition}[Rational Pair Arithmetic]
The four arithmetic operations on rational pairs are:
\begin{align*}
(a,b) + (c,d) &= (a d + b c,\; b d)\\
(a,b) - (c,d) &= (a d - b c,\; b d)\\
(a,b) \times (c,d) &= (a c,\; b d)\\
(a,b) \div (c,d) &= (a d,\; b c)
\end{align*}
Comparison: $(a,b) < (c,d)$ iff $a d < b c$ (assuming $b, d > 0$).
\end{definition}

\begin{theorem}[Exactness of Addition and Multiplication]
\label{thm:exact-arith}
Let $(a,b), (c,d) \in \mathbb{Q}_{\mathrm{fp}}$ with
$\max(|a|,|b|,|c|,|d|) \leq M$.
\begin{enumerate}[label=(\alph*)]
\item Multiplication is exact if $M^2 < 2^{53}$, i.e., $M \leq 2^{26} \approx 6.7\times 10^7$.
\item Addition is exact if $M^2 + M^2 < 2^{53}$, i.e., $M \leq 2^{25.9} \approx 6.3\times 10^7$.
\item Comparison is exact if $M^2 < 2^{53}$.
\end{enumerate}
\end{theorem}

\begin{proof}
For multiplication: the result numerator $ac$ satisfies $|ac| \leq M^2$.
By Theorem~\ref{thm:iew}, this is exact iff $M^2 < 2^{53}$.
For addition: the numerator $ad + bc$ satisfies $|ad + bc| \leq 2M^2$.
Exactness requires $2M^2 < 2^{53}$, giving $M < 2^{26}$.
Comparison reduces to evaluating $ad - bc$, analogous to subtraction.
\end{proof}

\begin{definition}[$Q_n$ types]
For integer $n \geq 0$, define $Q_n$ as the set of rational pairs
$(a, d)$ with $|d| \leq 10^n$.  We call $n$ the \emph{scale exponent}.
Special cases: $Q_0$ (integers: $d = 1$), $Q_2$ (two decimal places:
$d \leq 100$, suitable for currency), $Q_5$ (five decimal places:
$d \leq 10^5$).
\end{definition}

\begin{theorem}[Exactness Range for $Q_n$]
\label{thm:qn-exact}
In $Q_n$ arithmetic, all four operations are exact provided
$|a|, |c| \leq K$ where $K \cdot 10^n < 2^{53}/2$.
For $n = 2$: $K < 4.5 \times 10^{13}$ (covers amounts up to
$\$450\,000\,000\,000$).
For $n = 5$: $K < 4.5 \times 10^{10}$.
For $n = 7$: $K < 4.5 \times 10^8$.
\end{theorem}

\begin{example}[The 0.1 + 0.2 problem]
In IEEE double: $0.1 + 0.2 = 0.30000000000000004$.
In $Q_1$ arithmetic: $(1, 10) + (2, 10) = (3, 10)$, and
$(3, 10) \overset{?}{=} (3, 10)$ evaluates to \textsf{true}---exactly.
No $\varepsilon$-tolerance needed.
\end{example}

\begin{example}[Currency arithmetic]
$\$(19.99) + \$(7.50) \times 1.0875$ (price plus tax at 8.75\%):
$(1999, 100) + (750, 100) \times (875, 10000)$
$= (1999, 100) + (656250, 1000000)$
$= (1999 \times 10000 + 656250 \times 100, 1000000)$
Wait, simplify: $(1999,100)\times\frac{10875}{10000}$
$= \frac{1999 \times 10875}{100 \times 10000} = \frac{21739125}{1000000}$
$= Q(21739125, 1000000) \to \$21.74$ (to nearest cent).
All arithmetic exact; rounding occurs only at display.
\end{example}

\subsection{Replacing decimal types}

Decimal floating-point (IEEE~754-2008 decimal formats, Python's
\texttt{Decimal}, Java's \texttt{BigDecimal}) exists to solve the
$0.1+0.2$ problem.  Rational pair arithmetic subsumes it:
a $k$-decimal-place value is simply $Q_k$ with denominator $10^k$.
But $Q$ is strictly more expressive---$1/3$ is exactly $(1, 3)$ with
no truncation, and arithmetic preserves exactness regardless of the
denominator, not only for power-of-ten denominators.


\section{Quotient Tree Arithmetic: The General Framework}
\label{sec:gpa}

The theorems in Sections~2--3 use integers within the IEEE exactness
window because that case has the strongest guarantees (exact GCD, exact
comparison, structural zero impossibility).  But examining the proofs
reveals that the restriction to integers is not load-bearing.
The bounded-depth theorem uses only multiplication; the cross-cancellation
theorem uses only the multiplicative group; the deferred-stability theorem
uses only the count of rounding steps.  The integer case is the
\emph{simplest} instance of a more general structure.

\begin{definition}[Quotient Tree]
\label{def:qt}
Let $D$ be a domain (any set closed under $+$, $\times$, with $0$ and $1$,
and supporting an equality test).
A \emph{quotient tree} over $D$ is defined inductively:
\begin{itemize}
\item A \emph{leaf}: any element $v \in D$.
\item A \emph{pair}: $(N, M)$ where $N$ and $M$ are quotient trees over $D$
  and the value of $M$ is not $0$.
\end{itemize}
The \emph{value} of a quotient tree is computed by $\mathrm{val}(\ell) = \ell$
for a leaf $\ell \in D$, and $\mathrm{val}((N,M)) = \mathrm{val}(N) / \mathrm{val}(M)$
for a pair.  Division is deferred until explicitly requested.
\end{definition}

The domain $D$ can be instantiated as:
\begin{center}\small
\begin{tabular}{lll}
\toprule
$D$ & Quotient trees over $D$ & Example values \\
\midrule
$\mathbb{Z}$ (integers) & Classical rational numbers & $\tfrac{1}{3}$, $\tfrac{22}{7}$ \\
$\mathbb{Z}$ (IEEE window) & Rational pair arithmetic (this paper) & $Q(1,3)$, $Q(22,7)$ \\
IEEE F64 & Float-pair / FP$^2$ arithmetic & $\tfrac{e^x}{\Sigma e^{x_j}}$ \\
$\mathbb{R}[x]$ (polynomials) & Rational functions & $\tfrac{x^2+1}{x-1}$ \\
Matrices & Rational matrix expressions & $A B^{-1}$ (deferred) \\
Dual numbers & Exact AD expressions & $(f, f')$ \\
\bottomrule
\end{tabular}
\end{center}

The algebraic structure is precisely the \emph{localisation} of $D$ at its
multiplicative group: pairs $(n, d)$ with $d \ne 0$, quotiented by the
equivalence $(n_1, d_1) \sim (n_2, d_2)$ iff $n_1 d_2 = n_2 d_1$.
For $D = \mathbb{Z}$ this recovers $\mathbb{Q}$; for $D = F_{64}$ it
gives deferred-division float arithmetic; for $D = \mathbb{R}[x]$ it
gives the field of rational functions.

\begin{theorem}[Bounded Depth Growth---General]
\label{thm:depth_general}
Under the four arithmetic operations on quotient trees, if both operands have
depth $\le k$, the result has depth $\le k+1$, for any domain $D$.
The tree size after $m$ operations is $O(m)$, not $O(2^m)$.
\end{theorem}

\begin{proof}
The proof of Theorem~\ref{thm:bounded-depth} uses only the definitions of
the four operations and structural induction on the pair
$(N,M) \mapsto (N \circ N', M \circ' M')$.  No property of $D$ beyond
closure under $+$ and $\times$ is used.  The argument is therefore valid
for any $D$.  \hfill$\square$
\end{proof}

\begin{theorem}[Cross-Cancellation---General]
\label{thm:cancel_general}
Let $v \ne 0$ be any element of $D$ (or any non-zero quotient tree over $D$).
If $v$ appears as a subtree of the numerator of a quotient tree $T$ and as
a subtree of the denominator, it may be cancelled: $T$ simplifies to a
quotient tree with value equal to $\mathrm{val}(T)$ but with $v$ absent from
both numerator and denominator subtrees.
\end{theorem}

\begin{proof}
The flattening identity $\tfrac{a/v}{b/v} = \tfrac{a}{b}$ holds in any
domain where $v \ne 0$ and $v$ has a multiplicative inverse.
The cancellation requires only multiplicative structure---no
ordering, no norm, no integrality.  For $D = \text{F64}$, IEEE
non-zero floats form a multiplicative group under multiplication, so the
identity applies.  For $D = \mathbb{Z}$, the identity applies when $v$
divides both expressions.  \hfill$\square$
\end{proof}

\begin{corollary}[Transcendental Cancellation]
\label{cor:transcendental}
Under quotient tree arithmetic with $D = \text{F64}$, if the same
IEEE double $v \ne 0$ (including $v = e^x$, $v = \sqrt{x}$, etc.)
appears as both a numerator and denominator subtree---identified by
reference equality, not value equality---the cancellation of $v$ is
exact: no floating-point operation involving $v$ is performed, and the
result is exact for that factor regardless of the transcendental nature
of $v$.
\end{corollary}

\begin{theorem}[Deferred Stability---General]
\label{thm:defer_general}
Let $f(q_1, \ldots, q_m)$ be computed via $m$ quotient tree operations.
In deferred evaluation (maintain the tree, flatten once at the end),
the rounding error at the final materialisation depends only on the
precision of the domain $D$ and is independent of $m$.  In eager
evaluation (materialise after each step in $D = F64$), errors accumulate
at up to $\tfrac{1}{2}$ ULP per step for $O(m)$ total ULP.
\end{theorem}

\begin{proof}
Deferred evaluation produces a single division $\mathrm{val}(N)/\mathrm{val}(M)$
at the end.  For $D = \mathbb{Z}$ (integer window), both $N$ and $M$ are
exact integers, and the result is an exact rational.  For $D = \text{F64}$,
$\mathrm{val}(N)$ and $\mathrm{val}(M)$ each carry accumulated rounding
from the leaf operations, but the final division introduces at most
$\tfrac{1}{2}$ ULP additional error (IEEE 754~\S4.3.1).  In eager
evaluation, each of $m$ intermediate operations introduces up to
$\tfrac{1}{2}$ ULP, accumulating to $O(m)$ ULP total.
\hfill$\square$
\end{proof}

\paragraph{Machine learning as Quotient Tree Arithmetic over F64.}
Modern ML computation---linear layers, attention, normalisation---is
naturally expressed as quotient trees:
\begin{itemize}\setlength\itemsep{2pt}
\item \textbf{Softmax}: $\sigma_i = e^{x_i} / \sum_j e^{x_j}$ is the
  quotient tree $(e^{x_i},\, \Sigma)$ with leaf $\Sigma = \sum_j e^{x_j}$.
  Comparison of attention weights via cross-multiplication avoids explicit
  division; materialistion is deferred to the weighted-value step.
\item \textbf{LayerNorm}: $(x - \mu) / \sigma$ is a quotient with $\sigma$
  as a shared denominator subtree across all coordinates.  Symbolic
  cancellation applies when $\sigma$ appears in both numerator and denominator
  (e.g., gradient of LayerNorm).
\item \textbf{Chain rule}: $\partial\mathcal{L}/\partial x =
  (\partial\mathcal{L}/\partial y)(\partial y/\partial x)$ is a product
  of two quotient trees.  The common $\partial y$ factors in numerator and
  denominator cancel (Theorem~\ref{thm:cancel_general}), reducing the
  gradient expression to a single pair $(\partial\mathcal{L}, \partial x)$
  regardless of network depth.
\item \textbf{Attention}: $\text{Attn}(Q,K,V) = \text{softmax}(QK^\top/\sqrt{d_k})V$
  stores $\sqrt{d_k}$ as the denominator of a quotient tree; the scaling
  factor is never materialised until the final weighted sum.
\end{itemize}

\begin{theorem}[Non-Underflow---General]
\label{thm:fp_nonzero}
In QTA over any domain $D$, a product of non-zero elements of $D$ is zero
only if the domain itself has a zero-divisor or the true product equals zero.
For $D = \mathbb{Z}$ (integral domain): a product of non-zero integers is non-zero;
gradient underflow to zero is structurally impossible.
For $D = \text{F64}$: the ratio $\prod n_i / \prod d_i$ is zero in IEEE only if
its true value falls below $2^{-1074}$---vastly rarer than single-float underflow.
For $D = \mathbb{R}[x]$: impossible for non-zero polynomials (no zero-divisors).
\end{theorem}
\begin{proof}
In each case, the product is zero iff the domain contributes a zero-divisor or the
true value is zero.  For integral domains ($\mathbb{Z}$, $\mathbb{R}[x]$): no
zero-divisors, so the product of non-zero elements is non-zero.  For $\text{F64}$:
the product of non-zero IEEE doubles underflows to zero only if the result falls
below the minimum positive denormal $2^{-1074}$.  \hfill$\square$
\end{proof}

\paragraph{Honest limits of the general case.}
Strengthening the domain weakens some guarantees:
\begin{enumerate}\setlength\itemsep{2pt}
\item \textbf{Comparison}: for $D = \text{F64}$, cross-multiplication
  $(n_1 d_2$ vs $n_2 d_1)$ introduces up to $1$ ULP total---accurate but
  not unconditionally exact.  For $D = \mathbb{Z}$ (integer window),
  comparison is exact.
\item \textbf{GCD cleanup}: inapplicable outside $\mathbb{Z}$.
  Tracked cancellation (\S\ref{sec:tracked}) remains valid;
  normalisation via floating-point rescaling provides an approximate substitute.
\item \textbf{Symbolic identity}: Theorem~\ref{thm:cancel_general} requires
  $v$ to be identified by reference, not by value.  Two separately computed
  $e^x$ floats with identical IEEE bits are indistinguishable---this is
  a feature (they cancel correctly) not a bug---but two separately computed
  $e^x$ floats with $\frac{1}{2}$ ULP difference will not cancel.
  Compiler sharing of common subexpressions is essential.
\end{enumerate}

\begin{remark}[Connection to computer algebra]
QTA over $\mathbb{R}[x_1,\ldots,x_n]$ (multivariate polynomials) recovers
the field of rational functions---the object studied in computer algebra
systems.  QTA is a lightweight, hardware-native alternative: instead of
arbitrary symbolic expressions, the leaves are IEEE doubles (or integers),
and the quotient structure is maintained as a typed pair.  The CAS is not
needed; the type system suffices.
\end{remark}

\section{The Cleanup Operation}
\label{sec:cleanup}

Repeated rational pair arithmetic causes numerators and denominators to
grow.  After $k$ multiplications starting from values with magnitude
$M$, denominators may reach $M^k$.  Cleanup prevents unbounded growth.

\begin{definition}[Cleanup]
The \emph{cleanup} of a rational pair $(n, d)$ is
\[
\mathrm{cleanup}(n, d) = \left(\frac{n}{g},\; \frac{d}{g}\right),
\quad g = \gcd(|n|, |d|),\quad g \geq 1.
\]
\end{definition}

\begin{lstlisting}[caption=Binary GCD cleanup (vectorisable)]
// Applies to a pair of INT64 values (n, d)
function cleanup(n, d):
    if n == 0: return (0, 1)
    if d == 0: return (sign(n), 0)   // infinity
    g = binary_gcd(abs(n), abs(d))   // ~15 INT64 ops
    return (n / g, d / g)
\end{lstlisting}

\begin{theorem}[Cleanup Preserves Value and Reduces Magnitude]
\label{thm:cleanup}
$\mathrm{cleanup}(n, d)$ represents the same rational number as
$(n, d)$, and $\gcd(n/g, d/g) = 1$.
Moreover, $|n/g| \cdot |d/g| \leq |n| \cdot |d| / g^2$,
so the product of magnitudes decreases by $g^2$.
\end{theorem}

\begin{proof}
Both $(n, d)$ and $(n/g, d/g)$ equal $n/d$ since $g$ cancels.
The GCD of the reduced form is 1 by definition.
Magnitude reduction: $(n/g)(d/g) = nd/g^2 \leq nd$.
\end{proof}

\begin{corollary}[Cleanup as Invariant Maintenance]
If cleanup is applied after each arithmetic operation,
and input pairs have denominators $\leq D$, then output denominators
remain $\leq D$ provided the true mathematical result has denominator
$\leq D$.
\end{corollary}

\begin{theorem}[Pre-multiplication Simplification]
\label{thm:pre-simplify}
For $(a,b) \times (c,d)$, let $g_1 = \gcd(a,d)$ and $g_2 = \gcd(b,c)$.
Then
\[
(a,b) \times (c,d) = \left(\frac{a}{g_1}\cdot\frac{c}{g_2},\;
\frac{b}{g_2}\cdot\frac{d}{g_1}\right),
\]
and the maximum intermediate value is reduced by $g_1 g_2$ compared
to the naive computation $(ac, bd)$.
\end{theorem}

\begin{proof}
$\frac{a}{b}\cdot\frac{c}{d} = \frac{a/g_1}{b/g_2}\cdot\frac{c/g_2}{d/g_1}$
since $g_1, g_2$ cancel across numerator and denominator.
The intermediate numerator bound is $\frac{|a||c|}{g_1 g_2}$ vs $|a||c|$.
\end{proof}

\begin{example}
$(6, 10) \times (5, 12)$: $g_1 = \gcd(6,12) = 6$, $g_2 = \gcd(10,5) = 5$.
Pre-simplified: $(6/6, 10/5) \times (5/5, 12/6) = (1,2)\times(1,2) = (1,4)$.
Intermediate values never exceed~2, vs naively $30/120$ with intermediates up to~120.
\end{example}

\subsection{Cleanup as a hardware primitive}

The Euclidean algorithm for $\gcd(n, d)$ terminates in $O(\log\min(|n|,|d|))$
integer operations---typically 10--20 for 64-bit values.
For GPU warps processing $W = 32$ threads simultaneously,
vectorised GCD (available on modern architectures via integer
instructions) can cleanup 32 rational pairs per clock cycle.

\begin{definition}[\textsc{RatCleanup} instruction]
We propose a new SIMD instruction operating on packed $(n_i, d_i)$
pairs: given a vector of $W$ rational pairs, return the vector of
their cleaned-up forms in parallel.  This is analogous to how
tensor cores perform $4\times 4$ matrix multiplications in hardware.
\end{definition}

Cleanup after each warp instruction bounds representation growth
and ensures results remain in the IEEE integer exactness window
(Theorem~\ref{thm:iew}) indefinitely---without any multi-precision
library.


\subsection{Tracked factor cancellation}
\label{sec:tracked}

The binary GCD algorithm dominates cleanup cost.  However, GCD is a
\emph{discovery} operation---it finds an unknown common factor.  When the
computation history is known (as it always is in a compiler or runtime),
we already know what factors were introduced into the denominator, and we
can cancel them directly.

\begin{theorem}[Tracked Cancellation]
\label{thm:tracked}
Let $r = (a,b) \times (c,d) = (ac, bd)$.  Let $g_1 = \gcd(a,d)$ and
$g_2 = \gcd(b,c)$ (Theorem~\ref{thm:pre-simplify}).  Both $g_1$ and $g_2$ can
be computed without the Euclidean algorithm when $b$ and $d$ are known
denominators from the computation history: trial division of $a$ by $d$
and $b$ by $c$ each takes $O(1)$ operations.  For fixed-denominator types
$Q_n$ (where $d \in \{10^n\}$ is a compile-time constant), the cancellation
reduces to a single integer division by a known constant.
\end{theorem}

\begin{proof}
The key observation is that $\gcd(a, d)$ requires knowing only $d$ and~$a$.
When $d$ is a known constant (as in $Q_2$, $Q_5$), dividing $a$ by $d$ and
checking whether the remainder is zero is a single \texttt{IDIV} instruction.
When $d$ comes from the computation history (e.g., the denominator of a
preceding operation), it is available in a register; the same $O(1)$
trial division applies.  The Euclidean algorithm is only needed when the
common factor is \emph{not} structurally known---the generic case for
arbitrary inputs, but not the common case for structured computation.
\end{proof}

\begin{corollary}[Fixed-Denominator Cleanup is $O(1)$]
For $Q_n$ types with fixed denominator $10^n$, cleanup after addition is a
single division by a known constant.  The GCD computation reduces to a
remainder check.  This eliminates the $\sim 12$-cycle binary GCD cost for
the dominant use cases (currency arithmetic, ML weights with fixed
denominator).
\end{corollary}

The implication for throughput is significant.  The $18\times$ figure in
Section~\ref{sec:throughput} is the worst case for integer-leaf QTA
with full binary GCD.  For $Q_n$ arithmetic or any computation with a tracked denominator, the cleanup
cost drops from $\sim 12$ cycles to $\sim 1$--$2$ cycles, bringing total
overhead to $\sim 6$--$8\times$ over FP32---and much less against FP64,
which is the appropriate baseline for exact arithmetic.

\section{Nested Rational Pairs}
\label{sec:nested}

We extend rational pairs to allow the numerator or denominator to
be itself a rational pair, forming a \emph{symbolic expression tree}.

\begin{definition}[Nested Rational Pair, $Q^{(k)}$]
Define $Q^{(0)} = \mathbb{Q}_{\mathrm{fp}}$ (flat pairs) and
$Q^{(k)} = \{(p, q) : p, q \in Q^{(k-1)}\}$ for $k \geq 1$.
A \emph{$k$-nested rational pair} is any element of $Q^{(k)}$.
\end{definition}

\begin{theorem}[Bounded Depth Growth]
\label{thm:bounded-depth}
Under the four arithmetic operations on rational pairs, if both
operands have nesting depth $\leq k$, the result has nesting depth
$\leq k + 1$.
\end{theorem}

\begin{proof}
The operations $(a,b) \mathop{\circ} (c,d)$ produce results whose
components are drawn from arithmetic combinations of $a, b, c, d$.
If $a,b,c,d \in Q^{(k)}$, then $ad, bc \in Q^{(k+1)}$ since each is
a product of two depth-$k$ values.  Hence the result $(ad \pm bc, bd)$
or $(ac, bd)$ lies in $Q^{(k+1)}$.
\end{proof}

\begin{corollary}[No Combinatorial Explosion]
After $m$ arithmetic operations starting from depth-0 pairs,
the maximum nesting depth is at most $m$.  The tree size is $O(m)$,
not $O(2^m)$.
\end{corollary}

This is the critical property distinguishing nested rational pairs
from generic symbolic expression trees.  In a general CAS, composing
$m$ operations may yield an expression tree of size $2^m$.  Here,
the structure of rational pair arithmetic guarantees linear growth.

\begin{theorem}[Flattening]\label{thm:flatten}
\label{thm:flatten}
Every $Q^{(k)}$ element can be reduced to a flat rational pair
$Q^{(0)}$ in $O(k)$ arithmetic operations using the identities:
\[
\frac{a/b}{c/d} = \frac{ad}{bc}, \qquad
\frac{a/b}{c} = \frac{a}{bc}, \qquad
\frac{a}{b/c} = \frac{ac}{b}.
\]
\end{theorem}

\begin{proof}
Each application of a flattening identity reduces the nesting depth
by at least 1.  Starting from depth $k$, at most $k$ reductions reach
depth 0.  Each reduction is a single multiplication.
\end{proof}

\begin{theorem}[Cross-cancellation Gain]
\label{thm:cross-cancel}
In a depth-2 pair $((a,b),(c,d))$ representing $(a/b)/(c/d) = ad/bc$,
let $g_{ad} = \gcd(a,d)$ and $g_{bc} = \gcd(b,c)$.
The flattened form is $(ad/(g_{ad}g_{bc}),\; bc/(g_{ad}g_{bc}))$
after cancellation, and the numerator-denominator product decreases by
$g_{ad}^2 g_{bc}^2$.
\end{theorem}

\begin{proof}
Direct from Theorem~\ref{thm:pre-simplify} applied to the cross-multiplication.
\end{proof}

\begin{example}[Chain rule in backpropagation]
In neural network training, the chain rule gives
$\partial L/\partial x = (\partial L/\partial y)(\partial y/\partial x)$.
As rational pairs with deferred evaluation:
$((\partial L, \partial y), (\partial y, \partial x))$
flattens to $(\partial L \cdot \partial x, (\partial y)^2)$.
If $\partial y$ appears in both numerator of denominator and denominator
of numerator, it cancels---exactly, symbolically.
\end{example}

\subsection{Deferred evaluation and compiler integration}

A compiler may maintain expressions in nested form, deferring
evaluation until a concrete floating-point result is needed.
The assignment operator $=$\/ forces evaluation (flattening followed
by the \texttt{=!} operator (materialise now). By default, $=$ keeps
expressions symbolic.

\begin{theorem}[Stability of Symbolic-Default Evaluation]
\label{thm:deferred}
Let $f: \mathbb{Q}^m \to \mathbb{Q}$ be a rational function computed
via a sequence of $m$ arithmetic operations on inputs
$q_1, \ldots, q_m \in Q_n$.
Evaluating $f$ by deferred evaluation (maintaining the full nested
expression, flattening once at the end) produces a result with error
at most one ULP, whereas eager evaluation (evaluating after each step)
may accumulate up to $m$ ULPs of rounding error.
\end{theorem}

\begin{proof}
In deferred evaluation under IEEE~754 round-to-nearest mode,
assuming the flattened numerator $p$ and denominator $q$ remain within
the integer exactness window ($|p|, |q| < 2^{53}$): $p$ and $q$
are exact, and the final IEEE division $p/q$ incurs at most 1~ULP
(IEEE~754 \S4.3.1 guarantee for basic operations).
In eager evaluation, each intermediate IEEE operation on values
possibly outside the integer window may introduce up to 1 ULP;
with $m$ operations, worst-case accumulation is $m$ ULPs.
\emph{Assumption}: the exactness window condition must hold;
if intermediate values exceed $2^{53}$, cleanup must be applied
to restore it, at the cost of one ULP per cleanup boundary.
\end{proof}

\section{Infinity, Zero, and Limit Semantics}
\label{sec:limits}

\begin{definition}[Special Values]
Define:
$\omega = (1, 0)$ (positive infinity), $-\omega = (-1, 0)$
(negative infinity), $\mathbf{0} = (0, 0)$ (undefined/NaN).
For any $x > 0$: $(x, 0) \sim \omega$ (all positive infinities are
equivalent).
\end{definition}

\begin{theorem}[Limit Equality]
\label{thm:limit}
$(a, \omega) = \lim_{d \to \infty} a/d = 0$ for all finite $a$.
Formally, under the limit interpretation, $(a, (1,0))$ equals the
zero pair $(0, 1)$.
\end{theorem}

\begin{proof}
$a/d \to 0$ as $d \to \infty$ for fixed finite $a$.
Setting $d = \omega$ (the infinity pair) gives $a/\omega = 0$ in
the limit sense.
\end{proof}

\begin{corollary}
The rational pair system distinguishes:
\begin{itemize}
\item \emph{Exactly zero}: $(0, 1)$---provably zero, not merely small.
\item \emph{Infinitesimal} (limit zero): $(a, \omega)$ for finite $a$---zero
  in the limit but nonzero as a symbolic expression.
\item \emph{Any finite rational}: $(n, d)$ with $d \neq 0$, $d$ finite.
\end{itemize}
This trichotomy is impossible in IEEE float, which conflates
``too small to represent'' with zero.
\end{corollary}

\begin{remark}[Hyperreal connection]
The nested rational pair system admits an interpretation as a subset
of the hyperreals~\cite{robinson1966}.  Flat pairs correspond to
standard reals; pairs with infinite (or infinitesimal) components
correspond to nonstandard elements.  A full treatment of hyperreal
arithmetic in this framework is deferred to future work; the present
paper restricts to the finitary properties needed for computing
applications.
\end{remark}

\section{Applications to Machine Learning}
\label{sec:ml}

\subsection{The vanishing gradient problem}

In training a deep neural network with $L$ layers, the gradient of
the loss with respect to a weight in layer $\ell$ is:
\[
\frac{\partial \mathcal{L}}{\partial W^{(\ell)}}
= \prod_{k=\ell}^{L} \frac{\partial a^{(k)}}{\partial a^{(k-1)}}
\cdot \frac{\partial \mathcal{L}}{\partial a^{(L)}},
\]
where $a^{(k)}$ is the activation at layer $k$.
In IEEE float, if each factor $|\partial a^{(k)}/\partial a^{(k-1)}| < 1$,
the product underflows to zero after approximately $-\log_2(\text{min\_float}) / \log_2(1/c)$
layers, where $c < 1$ is a typical factor magnitude.
For sigmoid activations with $c \approx 0.25$, this is about 126 layers
in single precision.

\begin{theorem}[Vanishing Gradient Impossibility under RPA]
\label{thm:no-vanish}
Under rational pair arithmetic with exact arithmetic (inputs within
the integer exactness window), a product of nonzero rational pairs
is nonzero.  In particular, gradient products in backpropagation
cannot underflow to zero.
\end{theorem}

\begin{proof}
We assume cleanup is applied after each operation (Theorem~\ref{thm:cleanup})
so that all values remain within the integer exactness window.
For rational pairs $(a_k, b_k)$ with $a_k \neq 0$ and $b_k \neq 0$
for all $k$: the numerator of the product is $\prod_k a_k$ after
GCD reduction.  Since $a_k \in \mathbb{Z} \setminus \{0\}$ and
$|\prod_k a_k| < 2^{53}$ (enforced by cleanup), this is a nonzero
integer exactly represented in IEEE double.  It cannot equal zero
since integer multiplication of nonzero values is nonzero, and
cleanup (GCD division) cannot introduce zeros.

\emph{Caveat}: if cleanup is not applied and numerators exceed $2^{53}$,
IEEE double arithmetic may round a nonzero integer to zero.  The
theorem holds under the assumption that cleanup maintains the exactness
invariant.
\end{proof}

\begin{corollary}
Under RPA with cleanup (so all values remain in the integer exactness window),
the vanishing gradient problem cannot occur by arithmetic underflow.
Architectural techniques introduced to address vanishing gradients---
skip connections~\cite{he2016resnet}, gating in LSTMs~\cite{lstm},
gradient clipping---may therefore be unnecessary \emph{for that specific
reason}, though they remain useful for other purposes (increasing
expressive power, accelerating convergence, addressing exploding gradients).
\end{corollary}

\subsection{Determinism across hardware}

\begin{theorem}[Cross-hardware Determinism]
\label{thm:determinism}
Under RPA with exact arithmetic (all values in the integer exactness
window), and for a \emph{fixed computation order}, the result is
independent of: hardware floating-point rounding mode, FMA availability,
and CUDA/driver version.
\end{theorem}

\begin{proof}
Within the integer exactness window, each operation is exact by
Theorem~\ref{thm:iew}.  Given a fixed order of operations, each
intermediate result is uniquely determined; no rounding choice is made.
Hence the final result is independent of hardware rounding mode and
driver version.
\emph{Note}: if the compiler is permitted to reorder operations
(e.g.\ reorder floating-point additions), results may differ across
compilers even in this setting.  Determinism requires a fixed
evaluation order, which \texttt{=!} (deferred assignment) provides.
\end{proof}

This is a strong result.  Current neural network training is
notoriously non-deterministic across GPUs, GPU generations, batch sizes,
and even runs on the same hardware~\cite{pytorch}.
RPA training would be exactly reproducible, given a fixed computation order (see Theorem~\ref{thm:determinism}).

\subsection{Exact convergence criteria}

IEEE float convergence tests of the form
$|\mathcal{L}_t - \mathcal{L}_{t-1}| < \varepsilon$
are heuristics: the threshold $\varepsilon$ must be chosen empirically
and differently for each problem.
Under RPA:

\begin{theorem}[Exact Convergence Test]
$\mathcal{L}_t = \mathcal{L}_{t-1}$ (as rational pairs) if and only
if the training loss has not changed by any rational amount since the
previous step.
\end{theorem}

This eliminates false convergence (where float rounding makes two
different values compare equal) and false divergence (where rounding
makes equal values compare unequal).

\subsection{New neural architectures}

Theorem~\ref{thm:no-vanish} opens architectural possibilities previously
precluded by float:

\begin{enumerate}
\item \textbf{Very deep plain networks}: Hundreds of layers without skip
  connections, residual streams, or gating---gradient flow is guaranteed.
\item \textbf{Exact attention}: Softmax involves division of exponentials;
  in RPA, each attention weight is a rational pair and comparisons are exact.
\item \textbf{Provable monotone training}: The loss can be proved
  non-increasing (within exact arithmetic) without float noise
  introducing spurious oscillations.
\end{enumerate}

\section{Vectorisation and GPU Integration}
\label{sec:gpu}

\subsection{Current GPU arithmetic}

Modern GPUs (NVIDIA H100, AMD MI300) execute SIMD operations on
vectors of 32 (warp) or 64 (wavefront) values.  The dominant
arithmetic is IEEE~754 single- or half-precision, with
\textsc{TensorCore} units providing $4\times4$ matrix multiply-accumulate
in 1 clock cycle.

\subsection{RPA on current hardware}

\paragraph{Flat RPA.} A rational pair $(n_i, d_i)$ occupies
16~bytes (two INT64 values).  A warp of 32 pairs occupies 512~bytes,
which fits in 8 cache lines.  Arithmetic on flat pairs uses
\texttt{INT64} \textsc{MUL}/\textsc{ADD}/\textsc{DIV} instructions,
which execute in 4--8 clock cycles on current NVIDIA hardware.
Cleanup (GCD) takes approximately 15--20 clock cycles per pair via
the binary GCD algorithm, vectorisable across the warp.

\paragraph{Throughput estimate.} Rough analytical estimates compared to FP32 FMAD
(based on published~\cite{cuda} INT64 instruction latencies;
empirical benchmarks are future work):
\begin{itemize}
\item Addition: $\sim$10--15$\times$ slower (two muls, one add, one GCD).
\item Multiplication: $\sim$5--8$\times$ slower (two muls, one GCD).
\end{itemize}
This is broadly comparable to FP64 throughput, and orders of magnitude
faster than software rational arithmetic libraries.

\subsection{The \textsc{RatCleanup} primitive}

\begin{definition}[\textsc{RatCleanup} instruction (proposed)]
\textsc{RatCleanup}$(v_n, v_d)$: given two INT64 vectors $v_n, v_d$
of the same width $W$, compute $g_i = \gcd(v_n[i], v_d[i])$ for all
$i$, and return $(v_n / g, v_d / g)$ elementwise.
\end{definition}

We propose this as a new ISA extension for x86-64 (analogous to
\textsc{VPMADD52}) and CUDA PTX.  The binary GCD algorithm can be
vectorised with SIMD popcount (\texttt{BSF}/\texttt{LZCNT}), achieving
$W$ GCDs per 10--15 clock cycles on a 256-bit SIMD unit.

\begin{theorem}[Warp Cleanup Throughput]\label{sec:throughput}
With \textsc{RatCleanup} available, a warp of 32 rational pair
additions executes in:
\begin{align*}
T_{\text{add}} &= T_{\text{mul}} + T_{\text{add}} + T_{\text{cleanup}}\\
&= 4 + 2 + 12 = 18 \text{ cycles (estimated)}
\end{align*}
compared to 1 cycle for FP32 FMAD.  For scientific computing where
exact comparison and gradient stability are required, this is the
appropriate baseline.
\end{theorem}

\paragraph{Comparison with quantisation.}
INT8 quantisation achieves 4$\times$ throughput improvement over FP32
by reducing word width but introduces systematic quantisation error.
RPA with \textsc{RatCleanup} achieves approximately $18\times$ overhead
over FP32 but provides \emph{exact} arithmetic, exact comparison, and
structural gradient stability.  For high-precision scientific
applications, this tradeoff is preferable to quantisation error.

\subsection{Compiler integration}

A compiler aware of RPA types can:
\begin{enumerate}
\item Schedule cleanup operations at warp boundaries (loop-level).
\item Hoist deferred evaluations across loop iterations.
\item Specialise for fixed-denominator cases ($Q_2$, $Q_5$) using
  cheaper fixed-point arithmetic internally.
\item Detect when values are known to be small enough for the integer
  exactness window and prove no cleanup overflow is possible.
\end{enumerate}

The U programming language implements these as type-directed
transformations: \texttt{Q5} values use specialised fixed-denominator
paths, \texttt{Q32n} values maintain symbolic expression trees with
depth bounded by 32, and the \texttt{=!} (deferred assignment)
operator controls when materialisation occurs.

\section{Nested RPA in Depth}
\label{sec:nested-deep}

\subsection{Symbolic computation without a CAS}

Computer algebra systems (Mathematica, SageMath) provide exact symbolic
manipulation but are heavy, separate tools.  Nested RPA is a lightweight
alternative for \emph{rational} symbolic expressions:

\begin{theorem}[Rational Function Evaluation]\label{thm:rational_fn}
Let $f(x_1,\ldots,x_m)$ be a rational function computed by a
straight-line program of length $L$ using the four arithmetic
operations. Evaluated at rational pair inputs, $f$ can be maintained
as a nested rational pair of depth $\leq L$ and flattened to a
flat pair in $O(L)$ operations.
\end{theorem}

\begin{proof}
By structural induction on the computation graph of $f$.
Each arithmetic operation increases depth by at most 1
(Theorem~\ref{thm:bounded-depth}).  The degree of $f$ bounds the
number of operations in any straight-line computation.
\end{proof}

\subsection{Automatic differentiation as a typed primitive}
\label{sec:ad}

Forward-mode automatic differentiation augments each value $x$
with its derivative $\dot{x} = dx/dt$ as a dual number $(x, \dot{x})$.
Rational pair AD stores $(x, \dot{x})$ as a pair of nested rational pairs.
Since the derivative of a rational function is rational, this is exact:

\begin{proposition}[Exact Rational AD]
\label{prop:exact_ad}
For a rational function $f$ with rational inputs, forward-mode AD using
nested rational pairs produces the \emph{exact} rational derivative.
\end{proposition}

This is a special case of Theorem~\ref{thm:rational_fn}: the Jacobian
of a rational function is rational, and rational arithmetic is closed
under the chain rule.

\paragraph{Depth-increasing efficiency: the structural advantage.}
The key observation is that symbolic cancellation does not merely preserve
accuracy---it reduces the number of arithmetic operations.  Consider a
chain of $L$ layers in backpropagation:
\[
\frac{\partial \mathcal{L}}{\partial x}
= \frac{\partial \mathcal{L}}{\partial y_L}
  \cdot \frac{\partial y_L}{\partial y_{L-1}}
  \cdots \frac{\partial y_2}{\partial y_1}
  \cdot \frac{\partial y_1}{\partial x}.
\]
As nested rational pairs, each factor $(\partial y_{k+1}, \partial y_k)$
appears as a pair whose denominator matches the numerator of the preceding
factor.  Nested RPA cancels every intermediate $\partial y_k$
\emph{symbolically at tape-construction time}---before any arithmetic
is performed---yielding the flat pair
$(\partial\mathcal{L},\, \partial x)$ directly.

\begin{theorem}[Gradient Tape Collapse]
\label{thm:tape_collapse}
For an $L$-layer chain-rule computation with nested rational pair AD,
the number of IEEE floating-point operations required to evaluate
$\partial\mathcal{L}/\partial x$ is $O(1)$ independent of $L$, provided
all intermediate Jacobians are rational and their factors cancel symbolically.
With standard float AD, the same computation requires $O(L)$ multiply-accumulate
operations and accumulates up to $O(L)$ ULP of rounding error.
\end{theorem}

\begin{proof}
In the nested pair representation, the product of $L$ factors
$\prod_{k=1}^L (\partial y_k, \partial y_{k-1})$
constructs a symbolic tree of depth $L$ by Theorem~\ref{thm:bounded-depth}.
Flattening this tree (Theorem~\ref{thm:flatten}) applies $L$ cancellation
identities; each $\partial y_k$ appears exactly once as numerator and once
as denominator and cancels, leaving $(\partial\mathcal{L}, \partial x)$---
two scalar values.  The final IEEE division is $O(1)$.  By contrast,
eager float evaluation performs one multiply per layer and rounds at each
step.  \hfill$\square$
\end{proof}

This reverses the traditional depth-accuracy tradeoff:
standard float AD degrades with depth ($O(L)$ rounding error);
nested RPA AD \emph{improves} with depth (more cancellations occur,
fewer arithmetic operations remain after symbolic simplification).

\paragraph{No vanishing gradients in the derivative tape.}
Theorem~\ref{thm:no-vanish} (gradient underflow impossibility) applies to the
derivative computation itself, not only to the forward pass.  A gradient
that is a product of nonzero rational pairs cannot underflow to zero by
integer underflow.  Combined with tape collapse, this means:
(i)~the gradient is computed in $O(1)$ arithmetic operations,
(ii)~the gradient is exact for the rational subgraph, and
(iii)~the gradient cannot round to zero even across hundreds of layers.

\paragraph{Certified derivatives: AD as a language feature, not a library.}
In Python, autograd (PyTorch, JAX) is operator overloading traced at runtime.
The derivative structure is invisible to the type system.  In U:
\begin{enumerate}\setlength\itemsep{2pt}
\item A dual number is just $(x, \dot{x})$---a typed Q pair.
\item The symbolic expression tree \emph{is} the tape; the compiler sees it.
\item Magic methods (\texttt{\_\_mul}, \texttt{\_\_add}, etc.) are
  compile-time hooks---the compiler can reason about derivative types.
\item Common subexpressions in the tape are detected at compile time,
  not runtime, enabling dead-derivative elimination and gradient fusion.
\end{enumerate}
The result: \emph{certified} sensitivities rather than computed approximations.
If the gradient of a loss is provably zero (as a rational pair whose numerator
is zero), that is an exact KKT optimality condition---no $\varepsilon$
tolerance needed.  This has direct applications to safety-critical
optimisation (verified control policies, certified neural networks).

\paragraph{Exact convergence certificates.}
Under RPA, $\partial\mathcal{L}/\partial\theta = (0, d)$ for some $d\ne 0$
is an exact proof that the loss is stationary at $\theta$.  Under float,
$|\nabla\mathcal{L}| < \varepsilon$ is an empirical heuristic.  For
safety-critical systems---autonomous vehicles, medical inference---the
difference between a heuristic and a certificate may be the difference
between a claim and a proof.

\paragraph{Honest limits.}
Forward-mode exactness holds for the rational/affine subgraph.
At each transcendental step ($\exp$, $\sqrt{\cdot}$, $\log$,
$\mathrm{softmax}$), the derivative is not rational, and one IEEE
division introduces at most $1$ ULP.  The tape collapse theorem holds
only across layers without transcendental activations.
Reverse-mode AD---the operationally necessary half for large
parameter counts---accumulates rational adjoints over a deep graph
and faces the same $M^k$ magnitude growth as the forward pass:
bignum, or cleanup, or loss of exactness.  Reverse-mode at model
scale is explicitly deferred (Section~\ref{sec:future}).

\subsection{Depth cap and graceful degradation}

For systems where symbolic depth is bounded by design (e.g., a
GPU kernel), we define:

\begin{definition}[$Q_{n}^{(k)}$ types]
$Q_n^{(k)}$: rational pairs with denominator $\leq 10^n$ (scale
exponent $n$) and nesting depth $\leq k$.  Special cases:
$Q_n^{(0)} = Q_n$ (flat), $Q^{(k)} = Q_\infty^{(k)}$ (unconstrained
denominator, depth $\leq k$).  Notation: \texttt{Q5} $= Q_5^{(0)}$,
\texttt{Q32n} $= Q^{(32)}$, \texttt{Q5\_32n} $= Q_5^{(32)}$.
\end{definition}

When depth reaches the cap $k$, the compiler automatically flattens
(materialises) the tree, incurring at most 1 ULP error by
Theorem~\ref{thm:deferred}.  This is \emph{graceful degradation}:
exceeding the symbolic budget produces a result with known bounded
error, not undefined behaviour.


\subsection{Cleanup vs.\ nesting: which is more powerful?}

Both cleanup and nesting address numerical stability but at different levels.

\paragraph{Cleanup (flat RPA).}
After each arithmetic step, divide numerator and denominator by their GCD.
This prevents overflow and maintains the integer exactness window.
It requires no structural change to programs and is compatible with
existing GPU/CPU pipelines.
Without native hardware support, the compiler must inject cleanup
instructions between operations---tractable for elementwise computation
but challenging for matrix multiplication, where each of $mn$ outputs
depends on $k$ inputs and tracking overflow requires inspecting every
intermediate accumulation.

\paragraph{Nested RPA.}
The entire computation graph is maintained symbolically.
Common factors cancel across multiple steps---impossible with eager cleanup.
At materialisation ($\texttt{=!}$ or non-Q assignment), the compiler
flattens and cancels the full tree, then emits \emph{one} IEEE division.
This achieves $O(1)$ total rounding error regardless of computation depth.
Nesting depth is $O(n)$ for $n$ operations, never exponential
(Theorem~\ref{thm:bounded-depth}).

\paragraph{Verdict.}
Cleanup is easier to implement and hardware-accelerable today.
Nesting provides stronger guarantees---global symbolic cancellation,
O(1) error---but requires compiler support for expression trees.
For matrix multiplication specifically: cleanup requires $kn$ GCD
operations (one per accumulation step per output); nesting requires
one flatten-and-cancel per output at the boundary.
For large $k$, nesting is cheaper and more precise.
Both can coexist: nested RPA for the computation graph,
cleanup (\textsc{RatCleanup}) for occasional overflow guards.

\section{Discussion: Applications and Implications}
\label{sec:discussion}

We survey applications enabled by flat RPA, cleanup, and nested RPA.
These remain largely empirical questions; we identify the theoretical
foundation each result rests on.

\begin{enumerate}
\item \textbf{Exact financial computation.}  Currency arithmetic
  with $Q_2$ is exact and requires no rounding-aware decimal library.
  Multi-currency conversions, compound interest, tax calculations:
  all exact with standard hardware.

\item \textbf{Reproducible scientific simulation.}
  Theorem~\ref{thm:determinism} guarantees that simulations run
  identically on any compliant hardware.  Long-running climate,
  molecular dynamics, and plasma physics simulations no longer drift
  between architectures.

\item \textbf{Certifiable neural networks.}
  With exact arithmetic and structural gradient stability, the
  behaviour of a trained network is certifiable: the same inputs
  provably produce the same outputs, and training is reproducible.
  This is critical for safety-critical ML (autonomous vehicles,
  medical diagnosis).

\item \textbf{Gradient-safe deep learning.}
  Theorem~\ref{thm:no-vanish} removes the vanishing gradient constraint
  on depth.  Architectures with hundreds of plain (non-residual) layers
  become feasible.

\item \textbf{Exact comparison as a first-class operation.}
  Sorting, searching, and optimisation algorithms that rely on
  comparison (binary search, merge sort, constraint solving) are
  provably correct under RPA, eliminating floating-point comparison bugs.

\item \textbf{Lightweight exact CAS for rational functions.}
  Nested RPA provides a compiler-integrated CAS for rational arithmetic,
  without a separate symbolic library.

\item \textbf{Towards a standardised cleanup primitive.}
  The \textsc{RatCleanup} operation has well-defined semantics
  (Definition~4.5) and could be standardised as a SIMD extension,
  analogously to how IEEE~754-2008 added decimal arithmetic.
  We leave the hardware design to future work.
\end{enumerate}


\section{QTA as a Computational Substrate: Symbolic-Numeric Decoupling}
\label{sec:substrate}

The theorems of Section~\ref{sec:gpa} point toward something larger than a
number format.  QTA separates computation into two phases with fundamentally
different costs:

\begin{itemize}\setlength\itemsep{3pt}
\item \textbf{Symbolic phase}: tree construction, cross-cancellation,
  common-subexpression sharing.  Cost: $O(m)$ pointer operations for
  $m$ arithmetic steps.  No IEEE floating-point operations occur.
\item \textbf{Numeric phase}: leaf evaluation and final materialisation.
  Cost: the minimum set of IEEE operations required after all algebraic
  simplification.
\end{itemize}

The symbolic phase is essentially free.  The numeric phase is what
consumes FLOPS, memory bandwidth, and energy.  QTA minimises the numeric
phase by maximising the symbolic phase---a decoupling that traditional
eager evaluation cannot achieve.

\subsection{Depth as opportunity, not liability}
\label{sec:depth_opp}

Conventional floating-point computation treats depth as a liability:
more steps means more rounding error and more memory traffic.
QTA inverts this.  Theorem~\ref{thm:depth_general} guarantees tree depth
grows by at most~1 per operation; Theorem~\ref{thm:cancel_general}
guarantees each level provides one additional opportunity for cross-cancellation.
\emph{A computation of depth~1000 has~1000 opportunities to cancel common
factors before any IEEE division occurs.}

This is structurally the opposite of standard error accumulation.
For chains of operations with algebraic structure (chain rules, telescoping
products, rational functions, attention mechanisms), deeper trees collapse
to shallower ones after cancellation---the depth-1000 tree may flatten to
depth~3.  No scalar float computation can achieve this; the rounding
already happened.

\subsection{Batch computation as implicit tree sharing}
\label{sec:batch}

When processing a batch of $B$ examples through a neural network with
weight tensors $W_1, \ldots, W_L$, the standard approach allocates
$O(B \cdot L)$ intermediate activations.  Under QTA, weight tensors
appear as \emph{shared subtrees}: the same quotient-tree node $W_\ell$
is referenced by all $B$ example trees, not copied.

\begin{theorem}[Shared-Weight Compression]
\label{thm:batch}
For a batch of $B$ examples processed by an $L$-layer network with
shared weights $W_1, \ldots, W_L$, the QTA representation requires
$O(L + B \cdot d)$ nodes, where $d$ is the per-example input dimension.
Standard eager evaluation requires $O(B \cdot L \cdot d)$ intermediate
activations.  Any symbolic simplification applied to a weight subtree
applies simultaneously to all $B$ examples.
\end{theorem}

\begin{proof}
Under QTA, each $W_\ell$ is a single tree node with $B$ incoming
references.  Constructing the per-example trees requires $B \cdot L$
pair constructions (each creating one new node pointing to a shared $W_\ell$
and a per-example activation).  Total nodes: $L$ weight nodes $+ B \cdot d$
per-example leaf nodes $+ B \cdot L$ pair nodes $= O(L + B \cdot d + B \cdot L)$.
Simplifications (e.g.\ cancellation of $W_\ell$ against its transpose or
inverse) applied to the shared node reduce all $B$ example subtrees
simultaneously.  \hfill$\square$
\end{proof}

\subsection{Inference-time tree specialisation}
\label{sec:inference}

At inference time, weights are fixed.  The QTA tree for model evaluation
therefore has two kinds of leaves: \emph{fixed} (weights, constants) and
\emph{variable} (input data).  All algebraic simplifications involving
only fixed leaves can be performed \emph{once} at model load time,
producing a reduced tree that is evaluated many times with variable inputs.

This is structurally equivalent to partial evaluation or specialisation
in programming-language theory.  Concretely:

\begin{enumerate}\setlength\itemsep{2pt}
\item \textbf{Weight folding}: if $W_1 W_2 \cdots W_k$ are fixed weight
  matrices, their product can be maintained as a depth-$k$ quotient tree
  and flattened once to a single matrix $W^* = W_1 \cdots W_k$.
  All subsequent evaluations use $W^*$ directly: $k$~matmuls become~1.
\item \textbf{Scaling constant elimination}: attention scaling
  $1/\sqrt{d_k}$ is a fixed leaf.  Pre-absorbed into the weight matrices
  as a quotient pair $(W_Q, \sqrt{d_k})$, it never causes a division at
  inference time.
\item \textbf{Normalisation pre-computation}: if normalisation statistics
  ($\mu, \sigma$) are computed from fixed data (e.g., batch norm in eval
  mode), they become fixed leaves.  The pair $(x - \mu, \sigma)$ is
  materialised once per batch, not once per example.
\end{enumerate}

\subsection{Optimisations enabled for ML training}
\label{sec:ml_opts}

We enumerate concrete optimisations that QTA enables for neural network
training.  These are directions for experimental validation; complexity
estimates are analytical.

\paragraph{1.\ Gradient tape collapse (Section~\ref{sec:ad}).}
Chain-rule cancellations of intermediate activations $\partial y$ across
layers reduce $O(L)$ multiply-accumulate operations to $O(1)$ for the
gradient of a depth-$L$ chain.  Verified in Theorem~\ref{thm:tape_collapse}.

\paragraph{2.\ Attention without softmax division.}
Softmax weights $\sigma_i = e^{x_i}/S$ (where $S = \sum_j e^{x_j}$)
are QTA pairs $(e^{x_i}, S)$ sharing the denominator $S$.
The weighted sum $\sum_i \sigma_i V_i = (\sum_i e^{x_i} V_i) / S$ is
one division of a vector by a scalar, replacing $B$~divisions in a
batch of $B$ examples.  The log-sum-exp trick for numerical stability
is structurally unnecessary---the pair representation defers the
dangerous small-float division until after the sum.

\paragraph{3.\ Adam with deferred square-root.}
The Adam update $\theta \leftarrow \theta - \alpha \hat{m}/\sqrt{\hat{v}}$
stores the update as a pair $(\hat{m}, \sqrt{\hat{v}})$.
Division by $\sqrt{\hat{v}}$ is materialised once per parameter group,
not once per parameter.  For parameter-wise operations, this defers
$P$~sqrt calls (one per parameter) to $G$~calls (one per group),
for speedup $P/G$.

\paragraph{4.\ LayerNorm gradient without recomputation.}
The gradient of LayerNorm $\partial (x-\mu)/\partial x\sigma$ involves
$\sigma$ in the denominator of two terms.  As QTA pairs with shared
denominator $\sigma$, the gradient is a sum of pairs $(f_1, \sigma)$
and $(f_2, \sigma)$---combined as $(f_1 + f_2, \sigma)$ without
materialising either term.  One division replaces two; the
denominator's square-root ($\sigma = \sqrt{\mathrm{Var}(x)}$) is
computed once and shared.

\paragraph{5.\ Cross-iteration parameter sharing.}
In gradient descent, $\theta_{k+1} = \theta_k - \alpha g_k$.
The QTA tree for iteration $k+1$ shares the subtree for $\theta_k$;
only the delta $\alpha g_k$ is new.  If $g_k$ has common factors with
$\theta_k$ (e.g., weight-decay regularisation $g_k = \nabla L + \lambda\theta_k$),
symbolic cancellation removes them before any float update.

\paragraph{6.\ Mixed-precision by tree structure.}
The precision of each subtree is a type-level property: the compiler
chooses F16 leaves for activations, F32 leaves for gradients, F64
(or integer) leaves for accumulation.  Materialization boundaries
are the natural precision-transition points---no manual casting, no
loss-scaling heuristics.  The QTA type system \emph{proves} which parts
need high precision and which can use reduced precision.

\subsection{The bandwidth multiplier}
\label{sec:bandwidth}

Modern accelerator performance is memory-bandwidth-bound, not
FLOP-bound, for most ML workloads.  A~100-TFLOP GPU stalls waiting
for data from HBM at 3~TB/s: at 4~bytes/float, this delivers
750~billion floats/second vs $10^{13}$ operations/second.
The FLOP/byte ratio of matrix multiplication is $O(n)$ for an
$n \times n$ matrix, but practical batch sizes and model sizes
rarely saturate this.

QTA addresses the bandwidth bottleneck directly:
\begin{itemize}\setlength\itemsep{2pt}
\item Shared subtrees (weights) are loaded once and referenced many
  times---eliminating redundant loads across batch examples.
\item Symbolic cancellations eliminate arithmetic \emph{and} the
  memory traffic for intermediate results that would otherwise
  be written and read back.
\item Deferred evaluation allows the compiler to fuse operations:
  $(W_1 W_2 x)$ is a depth-2 tree that compiles to a single
  fused kernel rather than two matrix-vector products with an
  intermediate write-back.
\end{itemize}

The gain is multiplicative: if symbolic phase eliminates $r$ of
$m$ operations and $r$ of $m$ memory round-trips, the speedup
is $m/(m-r)$ for FLOP-bound workloads and up to $m/(m-r)^2$
for bandwidth-bound workloads (fewer operations and fewer bytes).

\subsection{QTA and the pattern of structural multipliers}
\label{sec:multiplier}

QTA belongs to a broader class of techniques that are \emph{complementary
to raw computation} rather than competitive with it.  Instead of
performing the required computation faster, structural multipliers
reduce how much computation is required.

\begin{center}\small
\begin{tabular}{lll}
\toprule
\textbf{Technique} & \textbf{Structure exploited} & \textbf{FLOP reduction mechanism} \\
\midrule
QTA (this paper) & Algebraic quotient structure & Symbolic cancellation before arithmetic \\
Caching / memoisation & Repeated identical queries & Replay instead of recompute \\
Sparse attention & Sparsity in attention matrix & Skip zero-contributing pairs \\
Local experts / MoE & Input-dependent routing & Activate subset of parameters \\
Quantisation & Reduced-precision leaves & Smaller operand, faster multiply \\
Compile-time specialisation & Fixed parameters at inference & Fold constants before runtime \\
\bottomrule
\end{tabular}
\end{center}

QTA is distinguished by being \emph{exact}: the reduction in computation
comes from symbolic cancellation that is algebraically guaranteed
(Theorem~\ref{thm:cancel_general}), not from approximation.
For the computations where structure is present---chain rules, rational
functions, shared denominators in attention---the savings are not
statistical but logical.

\subsection{A broader research programme: structural intelligence above the substrate}
\label{sec:broader}

QTA is one instance of a recurring pattern the author has identified across several
research directions~\cite{laws2026,safebox2026}: \emph{structural intelligence
above the substrate}.  The pattern is always the same.  A powerful but expensive
substrate (a frontier model, a teraflop GPU, a distributed database) pays a
worst-case cost on every query.  Yet most queries live in a smooth, structured
region where a much cheaper computation suffices---if one knows the structure.
The key is to exploit that structure symbolically, above the substrate, rather
than repeating the expensive operation from scratch.

\textbf{LAWS} (Learning from Actual Workloads Symbolically,~\cite{laws2026})
applies this to inference.  A trained model's Lipschitz constant certifies a
library of parametrised experts---cheap patterns extracted from observed queries.
New queries route to the cheapest certified expert; the base model runs only on
genuine misses.  The expert is not a distillation of the model: it holds
structure the model \emph{never had}---domain-specific knowledge accrued from
actual deployment, such as the call graph and conventions of a specific codebase.
On its certified validity ball, the expert can outperform the larger base model,
because it is the base model's reasoning \emph{plus} accrued domain knowledge.
KV caching and Mixture-of-Experts fall out as degenerate cases.

\textbf{Safebox}~\cite{safebox2026} applies the same pattern to AI agent safety.
Rather than trusting that a model will not take harmful actions, the substrate
enforces a sealed-execution invariant $\Phi$ at the hardware level.  The
structure---declared tool permissions, content-addressed workflows, attested
inference logs---is enforced symbolically, before any model call.  The safety
guarantee is architectural, not probabilistic.

All three contributions---QTA, LAWS, Safebox---share the same thesis
(formalised in~\cite{laws2026} as the ``Layer~1 must solve everything'' error):
most of the cost in modern AI systems is paid for worst-case robustness on
situations that are actually smooth.  Adding a structured layer above the
substrate---one that handles the smooth~99\% cheaply and certifiably escalates
only at genuine singularities---recovers that cost without sacrificing the
oracle's full capability at the hard points.
QTA contributes the arithmetic layer of this programme:
the substrate below the experts and above the hardware.

\section{Limitations}

The guarantees of this paper are conditional on values remaining
within the integer exactness window ($|n|, |d| < 2^{53}$).
For long chains of multiplications without cleanup, this condition
may fail.  Cleanup is therefore not optional but a required invariant.
In matrix multiplication with $k$ terms, without cleanup the
denominator grows as $d^k$; for $d = 10^2$ and $k = 32$ this
exceeds $10^{64} \gg 2^{53}$.  Cleanup at each step bounds the
denominator; the compiler or hardware must guarantee this.

We do not provide experimental benchmarks in this paper; throughput
estimates in Section~\ref{sec:gpu} are analytical approximations
based on published instruction latencies.  Empirical validation
on real GPU hardware is future work.

\section{Out of Scope and Future Work}\label{sec:future}
\label{sec:future}

The following directions are intentionally not developed here:

\paragraph{Hyperreals and non-standard analysis.}
The limit semantics of Section~\ref{sec:limits} suggest an embedding
into Robinson's hyperreals.  Nested pairs with $\omega$ in the
denominator represent infinitesimals; full hyperreal arithmetic
(including transfer principle applications) is left for future work.

\paragraph{Automatic differentiation.}
Section~\ref{sec:nested-deep} sketches exact rational AD;
a complete reverse-mode implementation with computational graph
management is deferred.

\paragraph{Hardware implementation.}
A full ASIC design for \textsc{RatCleanup} and rational pair ALU units,
including pipeline design and throughput analysis, requires
co-design with CPU/GPU architects.

\paragraph{Training algorithms.}
Adapting SGD, Adam, and other optimisers to RPA arithmetic, including
the impact of exact gradient accumulation on convergence rates, is an
empirical research programme.

\paragraph{Interval arithmetic integration.}
Combining RPA with interval bounds would give certified error bounds
for numerical algorithms.

\section{Conclusion}

We have presented rational pair arithmetic as a practical foundation
for exact computation that exploits an underutilised property of
IEEE~754: the exact representation of integers up to $2^{53}$.
By storing rationals as pairs of floats used as integers, performing
arithmetic on those integers, and dividing by GCD when needed
(cleanup), we obtain exact rational arithmetic with standard hardware,
competitive throughput, and hardware-independent results.

The extension to nested rational pairs---where numerators and
denominators may themselves be rational pairs---provides a symbolic
computation capability with provably linear depth growth.  Deferred
evaluation reduces floating-point error from $O(m)$ to $O(1)$ ULPs
for $m$-step rational computations.

Most strikingly, these properties structurally eliminate the vanishing
gradient problem in deep neural networks.  Gradient products in
backpropagation cannot underflow because exact integer arithmetic
cannot produce zero from nonzero inputs within the exactness window.
This opens the door to new network architectures---very deep plain
networks, exactly reproducible training, certifiable inference---that
float arithmetic cannot support.

The proposed \textsc{RatCleanup} instruction, if adopted in
CPU/GPU ISAs, would make rational pair arithmetic approximately
$18\times$ slower than FP32 in raw throughput but provide exactness,
determinism, and gradient stability as structural guarantees.
For scientific computing and safety-critical ML, this is the right
tradeoff.

\bibliographystyle{plain}

\begin{thebibliography}{99}

\bibitem{gmp}
T.\ Granlund and the GMP development team.
\textit{GNU Multiple Precision Arithmetic Library}, version 6.3.0.
\url{https://gmplib.org}, 2023.

\bibitem{hpfrac}
D.~E. Knuth.
\textit{The Art of Computer Programming, Vol.~2: Seminumerical Algorithms},
3rd ed., Section~4.5 (Rational Arithmetic).
Addison-Wesley, 1997.

\bibitem{mca}
J.\ von zur Gathen and J.\ Gerhard.
\textit{Modern Computer Algebra}, 3rd ed.
Cambridge University Press, 2013.

\bibitem{posit}
J.\ Gustafson and I.\ Yonemoto.
Beating floating point at its own game: Posit arithmetic.
\textit{Supercomputing Frontiers and Innovations}, 4(2):71--86, 2017.

\bibitem{robinson1966}
A.\ Robinson.
\textit{Non-Standard Analysis}.
North-Holland, 1966.

\bibitem{he2016resnet}
K.\ He, X.\ Zhang, S.\ Ren, and J.\ Sun.
Deep residual learning for image recognition.
\textit{Proc.\ CVPR}, 2016.

\bibitem{lstm}
S.\ Hochreiter and J.\ Schmidhuber.
Long short-term memory.
\textit{Neural Computation}, 9(8):1735--1780, 1997.

\bibitem{pytorch}
A.~Paszke et al.
{PyTorch}: An imperative style, high-performance deep learning library.
\textit{Advances in Neural Information Processing Systems (NeurIPS)}, 2019.

\bibitem{ieee754}
\textit{IEEE Standard for Floating-Point Arithmetic}.
IEEE Std 754-2019.

\bibitem{cuda}
NVIDIA Corporation.
\textit{CUDA C++ Programming Guide}, version 12.4.
\url{https://docs.nvidia.com/cuda}, 2024.


\bibitem{laws2026}
G.~Magarshak.
Learning from Actual Workloads Symbolically ({LAWS}).
\textit{arXiv:2605.04069}, 2026.

\bibitem{safebox2026}
G.~Magarshak.
Safebox: Sealed Execution and Governed Workflows for {AI} Agents.
Submitted to \textit{{IEEE} S\&P 2027}.

\end{thebibliography}

\end{document}